\documentclass{article}

\usepackage{microtype}
\usepackage{graphicx}
\usepackage{subfigure}
\usepackage{booktabs} 
\usepackage{enumitem}
\usepackage[hyphenbreaks]{breakurl}

\usepackage{hyperref}

\usepackage[preprint]{icml2024}

\usepackage[utf8]{inputenc} 
\usepackage[T1]{fontenc}    
\usepackage{hyperref}       
\usepackage{url}            
\usepackage{booktabs}       
\usepackage{amsfonts}       
\usepackage{amsthm}
\usepackage{nicefrac}       
\usepackage{microtype}      
\usepackage{xcolor}         
\usepackage{amsmath}
\usepackage[algo2e,ruled,noend]{algorithm2e}
\usepackage{graphicx}
\usepackage{mathtools}
\usepackage{float}

\usepackage{color, colortbl}
\definecolor{Gray}{gray}{0.76}
\usepackage{comment}

\usepackage[capitalize,noabbrev]{cleveref}

\theoremstyle{plain}
\newtheorem{theorem}{Theorem}

\theoremstyle{definition}

\theoremstyle{remark}

\icmltitlerunning{Unleashing Graph Partitioning for Large-Scale Nearest Neighbor Search}

\newcommand{\emp}[1]{{\textbf{\textit{#1}}}}

\usepackage{setspace}

\newcommand{\poly}{\operatorname{poly}}

\newcommand{\R}{\mathbb{R}}

\newcommand{\prb}[2]{\mathop{{\bf Pr}}_{#1}\left[ #2 \right]}

\newcommand{\MST}[1]{\operatorname{MST}}

\newcommand{\bPi}{\boldsymbol{\Pi}}

\usepackage{footnote}

\newcommand{\cH}{\mathcal{H}}

\newcommand{\cX}{\mathcal{X}}

\newcommand{\Partition}{\boldsymbol{\Pi}}

\SetCommentSty{textit}
\DontPrintSemicolon
\SetAlCapHSkip{0pt}
\SetAlgoLined

\SetKwFor{ParallelFor}{for}{do in parallel}{}
\SetKwFor{ParallelWhile}{while}{do in parallel}{}
\SetKwIF{If}{ElseIf}{Else}{if}{}{else if}{else}{end if}

\renewcommand{\epsilon}{\varepsilon}

\newcommand{\eps}{\epsilon}

\newcommand{\answerTODO}[1][]{\textcolor{red}{\bf [TODO]}}

\newcommand{\lshrouting}[0]{\textsc{hRt}\xspace}
\newcommand{\kmeansrouting}[0]{\textsc{kRt}\xspace}

\newcommand{\ranking}[0]{\textsc{MinDist}\xspace}
\newcommand{\voting}[0]{\textsc{Voting}\xspace}

\begin{document}

\twocolumn[
\icmltitle{Unleashing Graph Partitioning for Large-Scale Nearest Neighbor Search}



\icmlsetsymbol{equal}{*}

\begin{icmlauthorlist}
\icmlauthor{Lars Gottesbüren}{kit}
\icmlauthor{Laxman Dhulipala}{umd,goog}
\icmlauthor{Rajesh Jayaram}{goog}
\icmlauthor{Jakub Lacki}{goog}
\end{icmlauthorlist}

\icmlaffiliation{kit}{Karlsruhe Institute of Technology. Work done while the first author was a student researcher at Google Research}
\icmlaffiliation{umd}{University of Maryland}
\icmlaffiliation{goog}{Google Research}

\icmlcorrespondingauthor{Lars Gottesbüren}{lars.gottesbueren@kit.edu}

\icmlkeywords{Machine Learning, ICML}

\vskip 0.3in
]



\printAffiliationsAndNotice{}  

\begin{abstract}

We consider the fundamental problem of decomposing a large-scale approximate nearest neighbor search (ANNS) problem into smaller sub-problems. The goal is to partition the input points into neighborhood-preserving \emph{shards}, so that the nearest neighbors of any point are contained in only a few shards.
When a query arrives, a \emph{routing} algorithm is used to identify the shards which should be searched for its nearest neighbors.
This approach forms the backbone of distributed ANNS, where the dataset is so large that it must be split across multiple machines.



In this paper, we design simple and highly efficient routing methods, and prove strong theoretical guarantees on their performance. A crucial characteristic of our routing algorithms is that they are inherently modular, and can be used with \emph{any} partitioning method. This addresses a key drawback of prior approaches, where the routing algorithms are inextricably linked to their associated partitioning method.
In particular, our new routing methods enable the use of \emph{balanced graph partitioning}, which is a high-quality partitioning method without a naturally associated routing algorithm.
Thus, we provide the first methods for routing using balanced graph partitioning that are extremely fast to train, admit low latency, and achieve high recall.
We provide a comprehensive evaluation of our full partitioning and routing pipeline on billion-scale datasets, where it outperforms existing scalable partitioning methods by significant margins, achieving up to 2.14x higher QPS at 90\% recall$@10$ than the best competitor.


\end{abstract}

\section{Introduction}\label{sec:intro}
Nearest neighbor search is a fundamental algorithmic primitive that is employed extensively across several fields including computer vision, information retrieval, and machine learning \cite{shakhnarovich2008nearest}.
Given a set of points $P$ in a metric space $(\cX, d)$ (where $d \colon \cX \times \cX \mapsto \mathbb{R}_{\geq 0}$ is a distance function between elements of $\cX$), the goal is to design a data structure which can quickly retrieve from $P$ the $k$ closest points from any query point $q \in \cX$.
The typical quality measure of interest is \emph{recall}, which is defined as the fraction of true $k$ nearest neighbors found.

Solving the problem with perfect recall in high-dimensional space effectively requires a linear scan over $P$ in practice and in theory~\cite{andoni2017optimal}.
Therefore, most work has focused on finding \emph{approximate} nearest neighbors (ANNS).
The most successful methods are based on quantization~\cite{ivf-pq, scann} and pruning the search space using an index data structure to avoid exhaustive search.
Widely employed index data structures are k-d trees, k-means trees~\cite{k-means-tree, spann, flann}, graph-based indices such as DiskANN/Vamana~\cite{diskann} or HNSW~\cite{hnsw-journal} and flat inverted indices~\cite{indyk1998approximate, neural-lsh, pyramid, multi-probe-lsh, ivf-pq}.
For a comprehensive review we refer to one of the several surveys on ANNS~\cite{andoni2018approximate, li2019approximate, bruch2024foundations}.

While graph indices offer the best recall versus throughput trade-off, they are restricted to run on a single machine due to fine-grained exploration dependencies and thus communication requirements.
If a distributed system is needed, the currently best approach is to decompose the problem into a collection of smaller ANNS problems that fit in memory.
The decomposition starts by performing \emph{partitioning} -- splitting the set $P$ into some number of subsets (the \emph{shards}).
Then, to query for a point $q \in \cX$, a lightweight \emph{routing} algorithm is used to identify a small subset of the shards to search for the nearest neighbors of $q$.
The search within a shard is then performed using an in-memory solution such as a graph index if the shards are large (e.g., in distributed ANNS), or exhaustive search if the shards are small.
This approach is known as IVF and is also popular as an in-memory data structure~\cite{ivf-pq, ivf-hnsw, scann}.
Optimizing the partition is crucial to achieving good query performance, because it allows querying only a small subset of the shards. 

The partitioning and routing methods are closely connected, and, in some cases, inextricably intertwined.
As a result, many existing routing methods require a particular partitioning method to be used.
For example, consider partitioning via k-means clustering.
Each shard (a $k$-means cluster) is naturally represented by the cluster center.
The folklore \emp{center-based routing} algorithm identifies the shards whose centers are closest to the query point to be searched, i.e., solves a much smaller ANNS problem.


It has been recently observed that very high quality shards can be obtained using the following method by \cite{neural-lsh} based on \emph{balanced graph partitioning} (GP).
Let $G$ be a $k$-nearest neighbor ($k$-NN) graph of $P$.
That is, a graph whose vertices are points in $P$, in which each point has edges to its $k$ nearest neighbors in $P$.
The shards are computed by partitioning the vertices into roughly equal size sets, such that the number of edges which connect different sets is approximately minimized.
More formally, compute a partition $\Partition \colon V \mapsto [s]$ of $G$ into $s \in \mathbb{N}$ shards of size $|\Partition^{-1}(i)| \leq \frac{(1+\epsilon)|P|}{s} \, \, \forall
i \in [s]$, while minimizing $| \{ (u,v) \in E \mid \Partition(u) \neq \Partition(v) \}|$.
This attempts to put the maximum possible number of nearest neighbors of each vertex $v$ in the same shard as $v$.
While this approach delivers significantly improved recall, it comes with two main limitations, which have hindered its adoption in favor of the widely used k-means clustering.

\begin{enumerate}[topsep=1pt,itemsep=1pt,parsep=0pt,leftmargin=12pt]
	\item The shards do not admit natural geometric properties, such as convexity of shards computed by k-means clustering. As a result, there is no obvious method to route query points to shards. In fact, the best known routing method requires training a computationally expensive neural model~\cite{neural-lsh}.

	\item A $k$-NN graph is required to partition the pointset. This is computationally expensive and seemingly introduces a chicken and egg problem, as computing a $k$-NN graph requires solving the ANNS problem.
\end{enumerate}

In this paper, we show how to address the above limitations and enable the benefits of balanced graph partitioning for billion-scale ANNS problems.

\noindent
\textbf{Contribution 1: Fast, Inexact Graph Building.}
We empirically demonstrate that a very rough approximate $k$-NN graph built by recursively performing {\em dense ball-carving} suffices to obtain good quality partitions and query performance, and that such a coarse approximation can be computed quickly.


\noindent
\textbf{Contribution 2: Fast, Accurate, and Modular Combinatorial Routing.}
We devise two combinatorial routing methods which are fast, high quality, and can be used with \emph{any} partitioning method, in particular with graph partitioning.
Both adapt center-based routing to large shards, which is needed for distributed ANNS.
Our first and empirically stronger routing method is called \kmeansrouting. It is based on sub-clustering the points within each shard using hierarchical $k$-means.
We use either the tree representation of the clustering or HNSW to retrieve center points, routing the query to the shards associated with the closest retrieved points.
Our second method called \lshrouting is based on locality sensitive hashing (LSH). 
It works by locating the query point in the sorted ordering of LSH compound hashes, and inspecting nearby points to determine which shards to query.


\noindent
\textbf{Contribution 3: Theoretical Guarantees for Routing.} We provide a theoretical analysis of two variants of \lshrouting, and establish rigorous guarantees on their performance.
Specifically, we prove that each query point will be routed to a shard containing one of its approximate nearest neighbors with high probability.
To the best of our knowledge, these are the first theoretical guarantees for the routing step.


\noindent
\textbf{Contribution 4: Empirical Evaluation.}
In our extensive evaluation, we demonstrate that shards obtained via balanced graph partitioning attain 1.19x - 2.14x higher throughput than the best competing partitioning method.
We analyze the shards and find that the concentration of ground-truth neighbors in the top-ranked shard per query is significantly higher (up to +25\%).
Our routing methods are multiple orders of magnitude faster to train than the existing neural network based approaches~\cite{neural-lsh, fac} while obtaining similar or better recall at similar or lower routing time.
Specifically, our \kmeansrouting can be trained in half an hour on billion-point datasets, compared to multiple hours required by the neural network approaches on 1000x smaller million-point datasets where \kmeansrouting training takes under a second.


\section{Partitioning}\label{sec:sharding}
In this section, we present two improvements to the partitioning method of~\cite{neural-lsh}.
To speed up the $k$-NN graph construction we present a highly scalable approximate algorithm in Section~\ref{sec:approximate-graph-building}.
Moreover, we propose an algorithm to compute overlapping shards in Section~\ref{sec:overlapping}, to prevent losses in boundary regions.

\subsection{Approximate $k$-NN Graph-Building}\label{sec:approximate-graph-building}

To speed up $k$-NN graph building, we use a simple approximate algorithm based on recursive splitting with \emp{dense ball clusters}.
While the number of points is not sufficiently small for all-pairs comparisons, we split them using the following heuristic.
We sample a small number of pivots from the pointset, and assign each point to the cluster represented by its closest pivot.
The clusters are treated recursively, either splitting them again or generating edges for all pairs in the cluster.
The intuition is that top-$k$ neighbors will be clustered together with good probability. 
Improved graph quality is attained via independent repetitions, and assigning points to multiple closest pivots.
The latter can only be done a few times -- we do it once at the first recursive split -- otherwise the runtime would grow exponentially.

\begin{figure}
	\includegraphics[width=\linewidth]{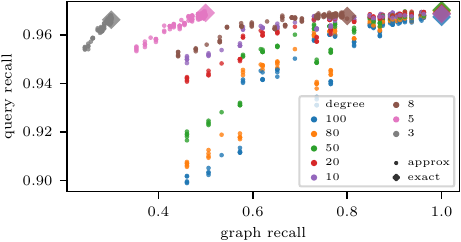}
	\vspace{-.7cm}
\caption{The $x$-axis shows the recall of the approximate $k$-NN graph used for graph partitioning. The $y$-axis shows the average recall of queries for $10$ nearest neighbors when only a single shard is inspected. The plot is computed using the SIFT1M dataset.
}\label{plot:graph_quality}
\end{figure}

Using a rough approximate $k$-NN graph for partitioning the dataset is acceptable since only coarse local structure must be captured to ensure partition quality, i.e., the edges need only connect near neighbors together, not necessarily the exact top-$k$.
Figure~\ref{plot:graph_quality} demonstrates that even low-quality graphs (graph recall $< 0.3$) lead to high query recall: more than 96\% of top-10 neighbors are concentrated in one shard per query.
Here we use $s=16$ shards with $|S_i| \leq 65625$ points per shard ($\epsilon=5\%$).
To obtain approximate graphs with different quality scores, we vary the number of repetitions and closest pivots, the cluster size threshold for all-pairs comparison, and the degree.

We also observe that using higher degrees may in some cases lead to slightly worse query recall. We suspect that since we use a low-quality method, adding more approximate neighbors pollutes the graph structure: note that this effect does not appear with an exact $k$-NN graph computed via all-pairs comparison.
Notably, the query performance between approximate and exact graphs differs only by a small margin.
Overall, these observations justify the use of highly sparse approximate graphs with only a few edges per point for the partitioning step, without sacrificing query performance.

\subsection{Partitioning into Overlapping Shards}\label{sec:overlapping}

When partitioning into disjoint shards, we can incur losses on points on the boundary between shards, whose $k$-NNs straddle multiple shards.
To address this issue, we propose a greedy algorithm inspired by local search for graph partitioning~\cite{fm} which eliminates cut edges by replicating nodes.
The set of cut ($k$-NN) edges is precisely the recall loss when the points themselves are queries~\cite{neural-lsh}, thus minimizing cut edges is a natural strategy in our setting.

We introduce a new parameter $o \geq 1$ to restrict the amount of replication.
For a fair comparison with disjoint shards in memory-constrained settings, we keep the maximum shard size $L_{\max}(s)$ fixed and instead increase the number of shards to $s' = o \cdot s$.
We first compute a disjoint partition into $s'$ shards of size $\frac{(1+\epsilon)|P|}{s'}$ and then apply an overlap algorithm with $\frac{(1+\epsilon)|P|}{s}$ as the final size.

In each step, our \emp{overlap algorithm} takes a node $u$ and places it in the shard $S_i$ that contains the plurality of its neighbors and not $u$.
This increases the average recall by $\frac{|\mathrm{cut}(u, S_i)|}{k |P|}$.
We repeat until there is no more placement into a below size-constraint shard which removes at least one edge from the cut.
We greedily select the node whose placement eliminates the most cut edges for the next step.

To parallelize this seemingly sequential procedure, we observe that nodes whose placement removes the same number of cut edges can be placed at the same time, i.e., grouped into bulk-synchronous rounds.
Only few rounds are necessary since each node removes at least one cut edge, no new cut edges are added, and $k$-NN graphs have small degree.

\section{Routing}\label{sec:routing}
\begin{figure}
\begin{center}
	\includegraphics[width=0.8\linewidth]{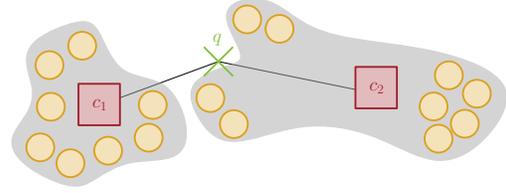}
	\end{center}
	\vspace*{-0.4cm}
	\caption{Illustration of an example where routing using a single center per shard fails. The nearest neighbors of $q$ are in the cluster of $c_2$, but $d(q, c_1) < d(q, c_2)$. If  the hierarchical sub-clusters are represented with their own centers, the routing works correctly.
	}\label{fig:aggregate-example}
\end{figure}

The goal of routing is to identify a \emph{small number of} shards which contain a large fraction of the $k$ nearest neighbors of a given query $q \in \cX$.

We adapt center-based routing for our purpose.
When using IVF for $|P| = 10^9$ points, a typical configuration uses $s \approx 10^6$ shards, with $|S_i| \approx 10^3$ points each~\cite{ivf-hnsw}.
Here, using a single center per shard works well. 
In contrast, in our setup we have few shards $s \in [10,100]$ of large size $|S_i| > 10^7$ and use a graph index per shard.

Large shards cannot be as easily represented with a single center.
Figure~\ref{fig:aggregate-example} illustrates the problem where routing fails because sub-cluster structures are not represented well by the center.
Moreover, the shards from graph partitioning are not connected convex regions as in $k$-means clustering, which further aggravates this issue.
For example, routing with a single center resulted in only 28\% recall in the top-ranked shard on the MS Turing dataset, as opposed to 81\% with our best approach.

Instead of a single center, we represent each shard $S_i \subset P$ by multiple points $R_i$ with the goal of accurately capturing sub-cluster structure.
At query time, given a query point $q \in \cX$ we retrieve from $R = \cup_{i \in [s]} R_i$ a set $Q$ of (approximately) closest points from $q$.
Then, we route the query to the shards which the points of $Q$ belong to.
More precisely, we compute a probe order of the shards, ranking each shard $i$ based on the distance between $q$ and the closest point in $Q \cap R_i$.
The size of $R$ is limited to a parameter $m$ so that the routing index fits in memory.

Our routing algorithms differ in the way the coarse representatives $R_i$ are constructed at training time and by the ANNS retrieval data structure used to determine $Q$.
The first algorithm uses hierarchical $k$-means to coarsen $P$ to $R$ and uses the resulting $k$-means tree for retrieval.
This method is called \emp{\kmeansrouting} for $k$-means RoutTing. In the experiments we also build an HNSW index on all points of $R$ for even faster retrieval.
This is called \emp{\kmeansrouting + \textsc{hnsw}}.
The second algorithm called \emp{\lshrouting} (HashRouTing) uses uniform random sampling to coarsen and a variant of LSH Forest~\cite{lsh-forest} to find $Q$.
In Section~\ref{sec:sortingLSH} we provide theoretical bounds for the quality of \lshrouting, and later demonstrate that both methods perform well in practice, with \kmeansrouting having a sizeable advantage.

\subsection{K-Means Tree Routing Index: \kmeansrouting}

\begin{algorithm2e}
\caption{\kmeansrouting: Training}\label{algo:tree:build}
\KwIn{Shards $S_i$, number of centroids $l$, index size $m$, cluster size $\lambda$}
\ParallelFor() {$i=1$ to $s$}{
    create root node $v_i$\;
    \FuncSty{Build}$(S_i, l, \frac{|S_i|(m-s)}{|P|}, \lambda, v_i)$\;
}
\SetKwFunction{Build}{Build}
\SetKwProg{Fn}{func}{:}{}
\Fn{\Build{$P$, $l$, $m$, $\lambda$, $v$}}{
    \lIf() {$m \leq 1$} { \KwRet }
    $\mathrm{centroids}(v) \gets \FuncSty{K-Means}(P, l)$ \;
    \ParallelFor() {$c \in \mathrm{centroids(v)}$} {
        $P_c \gets$ $k$-means cluster of $c$ \;
        create new tree-node $v_c$\;
        add tree-edge from parent $v$ to $v_c$ \;
        \If() {$|P_c| > \lambda$}{
            \FuncSty{Build}$(P_c, l, \frac{(m - l)|P_c|}{|P|}, \lambda, v_c)$
        }
    }
}

\end{algorithm2e}

\begin{algorithm2e}
\caption{\kmeansrouting: Routing}\label{algo:tree:routing}
\KwIn{Search budget $b$, query point $q$}
$\mathrm{PQ} \gets \{ (v_i, 0, i) \mid i \in [s] \}$ \tcp*[h]{min-heap with (tree node, key, shard ID)} prioritized by key\;
min-dist$[i] \gets \infty \quad \forall i \in [s]$ \;
\While() {$\mathrm{PQ}$ not empty and $b {--} > 0$} {
    $(v, d_v, s_v) \gets \mathrm{PQ}\FuncSty{.pop()}$\;
    \For() {$c \in \mathrm{centroids}(v)$}{
    	min-dist$[s_v] \gets \min(\text{min-dist}[s_v], d(q,c))$ \;
        \If() {$c$ has a sub-tree} {
            add $(v_c, d(q,c), s_v)$ to $\mathrm{PQ}$ \;
        }
    }
}
\KwRet sort $[s]$ by min-dist \;
\end{algorithm2e}

Algorithm~\ref{algo:tree:build} shows the training stage for \kmeansrouting.
We create one root node $v_i$ per shard $S_i$ and hierarchically construct a separate $k$-means tree for each shard.
The set $R_i$ consists of centroids across all recursion levels of hierachical $k$-means.
Per level we use $l=64$ centroids for the sub-tree roots.
Furthermore, we are given a maximum index size $m$, which we split proportionately among sub-trees according to their cluster size.
We subtract $l$ from the budget on each level, to account for the centroids.
The recursive coarsening stops once the budget $m$ is exceeded or the number of points is below a cluster size threshold $\lambda$ (we use 350).
In contrast to usual $k$-means search trees, we do not store or search the points in leaf-nodes, as our goal is to coarsen the dataset.

Algorithm~\ref{algo:tree:routing} shows the routing algorithm which is similar to beam-search~\cite{norvig-russell}.
At a non-leaf node we score the centroids against the query $q$ and insert the non-leaf children with the associated centroid distance into a priority queue for further exploration.
Initially the priority queue contains the tree roots.
The search terminates when either the priority queue is empty or a search budget $b$ (a parameter) is exceeded.

\subsection{Sorting-LSH Routing Index: \lshrouting}\label{sec:sortingLSH}

We now describe a routing scheme based on \textit{Locality Sensitive Hashing} (LSH). At a high level, an LSH family $\cH$ is a family of hash functions of the form $h:\cX \to \{0,1\}$, such that similar points are more likely to collide; namely, the probability $\prb{h \sim \cH}{h(x) = h(y)}$ should be large when $x,y \in \cX$ are similar, and small when $x,y$ are farther apart. We formalize this in the proof of Theorem \ref{thm:SortLSHMain}, and describe the routing index assuming we have an LSH family.


We now describe the construction of a SortingLSH index.
We first subsample $m$ points $R$ from $P$ uniformly at random. A single SortingLSH index is created as follows: we hash each point $x \in R$ multiple times via independent hash functions from $\cH$, and concatenate the hashes into a string $(h_1(x),h_2(x),\dots,h_t(x))$. Importantly, we use the same hash functions $h_1,\dots,h_t$ for each point in $R$. We \textit{sort} the points in $R$ \textit{lexicographically based on these strings of hashes}. The index is the points stored in sorted order along with their hashes. Intuitively, similar points are more likely to collide often, thus share a longer prefix in their compound hash string, and thus are more likely to be closer together in the sorted order. We repeat the process $r$ times to improve retrieval quality (up to 24 in our experiments).


The routing procedure works as follows. Given a query point $q$, for each of the $r$ indices, we hash $q$ with the LSH functions used for that repetition, compute the compound hash string $h(q) = (h_1(q),\dots,h_t(q))$ for $q$, and then find the position $\tau \in [m]$ that is lexicographically closest to $h(q)$.
We retrieve the window $[\tau - W, \tau + W]$ in the sorted order to consider these points' distances to the query in the ranking.
The pseudocodes for index construction and routing are given in Appendix~\ref{app:codes}.

One key advantage of employing LSH is that we can prove formal guarantees for the retrieved points.
Since our approach is more involved than searching through a single set of hash buckets, we provide a new analysis which demonstrates that it provably recovers a set of relevant points $Q$ that results in routing to a shard which contains an approximate nearest neighbor. We remark that, while the routing only guarantees we are routed to a shard containing an \textit{approximate} nearest neighbor, under mild assumptions on the partition, most of the approximate nearest neighbors of a given point will be in the same shard as the closest point. Under such a natural assumption, our SortingLSH index also recovers the true nearest neighbors.

\begin{theorem}\label{thm:SortLSHMain}
	Fix any approximation factor $c>1$, and let $P \subset (\R^d,\|\cdot\|_\rho)$ be a subset of the $d$-dimensional space equipped with the $\ell_\rho$ norm, for any $\rho \in [1,2]$. Set stretch factor $\alpha = O(c)$, repetitions $r = O(n^{1/c})$ and window size $W = O(1)$.
	Then the following holds for any query $q \in \R^d$:
	(a) If $m=n$, then with probability $1-1/\poly(n)$ $q$ gets routed to a shard $\bPi(p)$ containing at least one point $p \in P$ that satisfies
		$ \|p - q \|_\rho \leq \alpha \pi_1(q) $.
	(b) If $m<n$, then for any $\delta \in (0,1)$, with probability $1-\delta$ the query $q$ gets routed to a shard $\bPi(p)$ containing at least one point $p \in P$ that satisfies
		$ \|p - q \|_\rho \leq \alpha \cdot \pi_{\lceil \log\delta^{-1}  \frac{n}{m}\rceil  }(q) $
	. Here, $\pi_k(q)$ is the distance from $q$ to the $k$-th nearest neighbor of $q$ in $P$.

\end{theorem}


In addition to routing to the closest shard by distance, we also investigated a natural scheme where points in $Q$ vote for their shard with voting strength proportional to their distance.
This is helpful in scenarios where the majority of nearest neighbors are concentrated in one shard, but the top-1 neighbor is located in another.
Since the voting scheme did not achieve better empirical results, we present the details in~\cref{sec:appendix:voting}.

\section{Empirical Evaluation}\label{sec:experiments}
\begin{figure*}[t]
	\includegraphics[width=\textwidth]{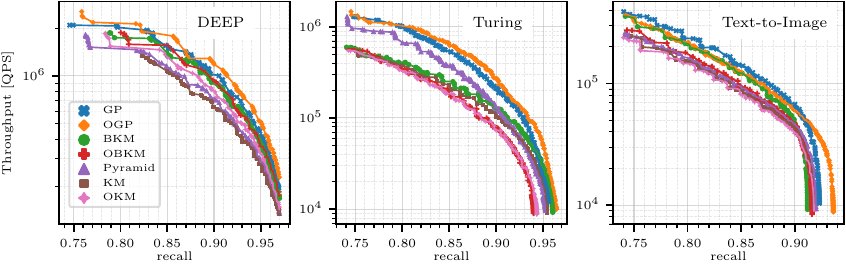}
	\vspace*{-0.8cm}
	\caption{Throughput vs recall evaluation on big-ann-benchmarks.
	}\label{fig:qps}
\end{figure*}

\newcommand{\nshards}{\eta}

We evaluate the performance of IVF algorithms for approximate nearest neighbor search using different partitioning and routing algorithms in terms of the recall of retrieving top-10 nearest neighbors versus number of shards probed $\nshards$ and throughput (queries-per-second).
For $k$-NN graph-building we also use $k=10$.
We test on the datasets \emph{DEEP} ($\mathrm{dim}=96$, $\ell_2$), \emph{Text-to-Image} ($\mathrm{dim}=200$, inner product, abbreviated as TTI) and \emph{Turing} ($\mathrm{dim}=100$, $\ell_2$) from big-ann-benchmarks~\cite{big-ann-benchmarks}, which have one billion points each with real-valued coordinates.
These are some of the {\em most challenging billion-scale ANNS datasets} that are currently being evaluated by the community. TTI (which uses a dual-encoder model) even exhibits out-of-distribution query characteristics, as demonstrated in~\cite{ood-diskann}.

We compare graph partitioning (GP) with three prior scalable partitioning methods: $k$-means (KM) clustering, balanced $k$-means (BKM) clustering~\cite{de2023balanced}, and Pyramid~\cite{pyramid}.
We allow $\epsilon=5\%$ imbalance of shard sizes for all these algorithms.
We also consider overlapping partitioning with 20\% replication ($o=1.2$) and prefix the corresponding method name with O (OGP, OKM, OBKM).
OKM and OBKM create overlap by assigning points to second-closest clusters (etc.) as proposed by~\cite{spann}.
We use $40$ shards in the non-overlapping case and $48$ shards in the overlapping case (so that the maximum shard size is the same in both cases).

We implemented our algorithms and baselines in a common framework in C++ for a fair comparison.
Our source code is available at \url{https://github.com/larsgottesbueren/gp-ann}.
We note that the original source code for Pyramid \cite{pyramid} is not available, and the original source code for BKM~\cite{de2023balanced} is not parallel, which is prohibitive for us.
To compute graph partitions, we use KaMinPar by~\cite{kaminpar} which is the currently fastest algorithm.
For an overview of the field of graph partitioning, we refer to a recent survey~\cite{gp-survey}.
The query experiments are run on a 64-core (1 socket) EPYC 7702P (2GHz 1TB RAM), and the partitioning experiments are run on a 128-core (2 sockets, 16 NUMA nodes) EPYC 7713 (1.5GHz 2TB RAM).

Unless mentioned otherwise, all methods use \kmeansrouting + HNSW for routing, since it is also the generalization of $k$-means' native routing method for large shards.
Refer to Appendix~\ref{appendix:baselines} for further algorithmic details of the baselines.
Due to space constraints, we largely defer discussion of tuning parameters in the main text to Appendix~\ref{appendix:configuration}.

\subsection{Large-Scale Throughput Evaluation}

We start with an evaluation of recall vs throughput (queries per second) of approximate nearest neighbor queries.
Due to resource constraints we simulate distributed execution on a single machine, processing hosts one after another.
While this setup does not take into account networking latency, we note that its impact should be negligible:  communication latency in modern networks is less than 1 microsecond~\cite{InfiniBand} whereas HNSW search latency is around 1 millisecond or higher.
Moreover, our algorithms (and all IVF algorithms in general) perform the least amount of communication possible -- only the query vector as well as neighbor IDs and distances are transmitted.

We assume a distributed architecture where each machine hosts one shard and the routing index.
For routing, queries are distributed evenly.
The machine that receives a query forwards it to the  hosts that are supposed to be probed.

We use HNSW to search in the shards and to accelerate routing on the \kmeansrouting points.
In total we use $60$ hosts.
This is larger than the number of shards, as we replicate popular shards to counteract query load imbalance.
We process queries on 32 cores per host in parallel.
Throughput is calculated based on the maximum runtime of any host. Query routing is distributed evenly across hosts, whereas in-shard searches are only accounted for the queries that probe the shard on the host. Replicas of the same shard evenly distribute work amongst themselves.


To obtain different throughput-recall trade-offs, we try different configurations of \kmeansrouting (index size $m$) and HNSW (ef\_search), as well as number of shards probed $\nshards$ or the probe filter methods proposed by~\cite{pyramid} and~\cite{spann}, which determine a different $\nshards$ for each query.
For Pyramid, we included its native routing method.
We only plot the Pareto-optimal configurations, which is standard practice~\cite{douze2024faiss}.

\begin{table}
	\caption{Throughput in $10^3$ queries per second at $0.9$ recall.}\label{tbl:qps}
	\vspace{0.22cm}
	\resizebox{\linewidth}{!}
	{
		\begin{tabular}{l rrrrr}
			\toprule
			QPS $[\cdot 10^3]$ & GP      & OGP    & BKM   & OBKM & Pyramid      \\
			\midrule
			DEEP             & 1002.5    & \bf{1090.2} & 868.3 & 919  & 713.5   \\
			Turing           & 208.9     & \bf{271.9}  & 124   & 78   & 127.3   \\
			TTI              & \bf{67.4} & 59.9        & 46.2  & 43.9 & 45.3    \\
			\bottomrule
		\end{tabular}
	}
\end{table}

\begin{table}
	\caption{Partitioning times in minutes. GB = graph-building.}\label{tbl:partitioning-times}
	\vspace{0.22cm}
	\resizebox{\linewidth}{!}
	{
		\begin{tabular}{l | rrr | rr | r}
			\toprule
			& GB     &  GP   & OGP         & BKM      & OBKM    & Pyramid \\
			\midrule
			DEEP             & 63m    &  24m  & +10m         & 24m     & +16m       & 32m   \\
			Turing           & 69m    &  17m  & +11m         & 36m     & +18m       & 35m   \\
			TTI              & 107m   &  17m  & +10m         & 65m     & +34m       & 60m    \\
			\bottomrule
		\end{tabular}
	}
\end{table}

Figure~\ref{fig:qps} shows the recall vs throughput plot.
Additionally, Table~\ref{tbl:qps} shows the throughput for a fixed recall value of $0.9$.
GP outperforms all baselines, and does so by a large margin on Turing and TTI.
Adding overlap with OGP improves results further across the whole Pareto front on DEEP and Turing.
On Text-to-Image, OGP loses to GP on recall below $0.91$, but OGP can achieve higher maximum recall.
Note that the maximum recall is constrained by the in-shard HNSW configurations.
Overall GP or OGP improves QPS over the next best competitor at 90\% recall by 1.19x, 1.46x, and 2.14x respectively on DEEP, TTI and Turing.


\begin{figure*}[t]
	\includegraphics[width=\linewidth]{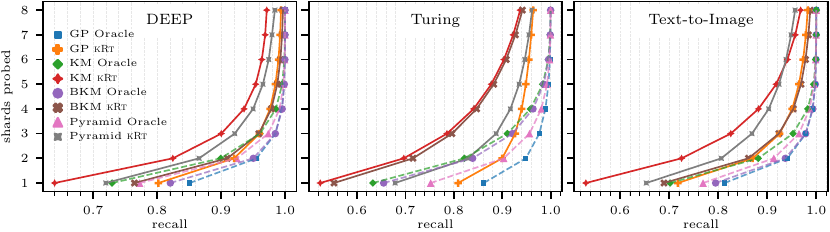}
	\vspace*{-0.8cm}
	\caption{Evaluating disjoint partitioning methods with \kmeansrouting and routing oracle (dashed), assuming exhaustive search in the shards.
	}\label{fig:compare-partitions}
\end{figure*}

\begin{figure}
	\includegraphics[width=\linewidth]{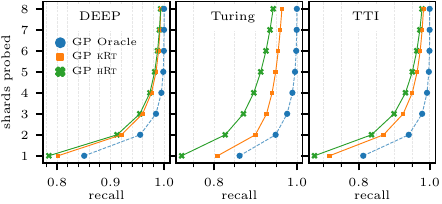}
	\vspace*{-0.8cm}
	\caption{\lshrouting vs \kmeansrouting with GP as the partition.}
	\label{fig:oracle-lsh}
\end{figure}

\begin{figure}
	\includegraphics[width=\linewidth]{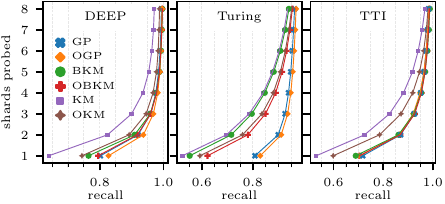}
	\vspace*{-0.6cm}
	\caption{Evaluating overlapping vs disjoint partitions with \kmeansrouting routing and exhaustive shard search.}
	\label{fig:compare-overlapping-partitions}
\end{figure}

\begin{figure}
	\includegraphics[width=\linewidth]{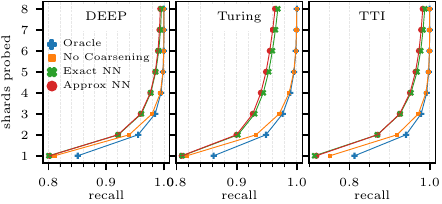}
	\vspace*{-0.6cm}
	\caption{Ablation study analyzing the sources of routing losses. Different curves correspond to different variants of a routing algorithm used. We note that all but \emph{Approx NN} are hypothetical and/or infeasible in practice.}
	\label{fig:analyzing-losses}
\end{figure}

\begin{figure}
	\includegraphics[width=\linewidth]{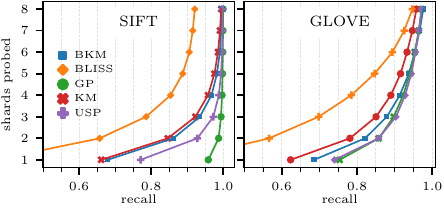}
	\vspace*{-.8cm}
	\caption{Comparison against neural network baselines.
	}\label{fig:small-scale}
\end{figure}

\subsection{Training Time}

In Table~\ref{tbl:partitioning-times} we report partitioning times.
The query performance of GP does come at the cost of a moderately increased partitioning time compared to the baselines.
This is due to the graph-building step (GB).
While previous works~\cite{neural-lsh, bliss, fac} identified graph partitioning as the bottleneck, we observed that it can be made very fast by using the KaMinPar~\cite{kaminpar} partitioner.
The overall partitioning times are quite moderate for datasets of this scale, and our implementations of the baselines are competitive.
The Pyramid paper reports a partitioning time of 87 minutes on a 500M sample of DEEP, demonstrating that our reimplementation is faster.
\cite{spann-build-times} report 4-5 days training time for their hierarchical balanced $k$-means tree SPANN.

While our graph-building implementation is already optimized, it can be further accelerated using GPUs~\cite{faiss-gpu} or more hardware-friendly implementation of bulk distance computations~\cite{scann}.

Coarsening with \kmeansrouting takes roughly 1000s on Turing, 800s on DEEP, and 2000s on TTI.
HNSW training in a shard of roughly 25M points takes 600s-1500s for Turing and DEEP, and 800s-1900s for TTI.
Note that these numbers apply to all baseline partitioners, as they use the same routing algorithms and in-shard search.
Training \lshrouting takes under 20 seconds.

\subsection{Analyzing Partitioning and Routing Quality}\label{sec:experiments:perfect-inshard-query}

To further assess the quality of the partitions and routing, we study recall independent of HNSW searches in the shards.
Specifically, we look at the number of probed shards versus the recall with exhaustive search in the shards, in order to assess the quality of the routing algorithm.
All recall values reported in Figures~\ref{fig:compare-partitions} to Figure~\ref{fig:analyzing-losses} and Section~\ref{sec:experiments:perfect-inshard-query} assume exhaustive search.
Furthermore, to assess the quality of the partition alone and provide context for routing, we look at recall with a \emph{hypothetical} routing oracle that knows the entire dataset and picks an optimal sequence of shards to probe, i.e., one that maximizes recall.

\noindent\textbf{Disjoint Partitions.}
Figure~\ref{fig:compare-partitions} shows the results for disjoint partitions with the oracle marked as dashed lines.
GP outperforms Pyramid, KM and BKM on all three datasets.
In particular on Turing the margin of improvement is very significant: $+25\%$ recall over BKM and $+13.4\%$ over Pyramid at $\nshards = 1$ probes with \kmeansrouting and $+11\%$, respectively $+20\%$ with the oracle.

In Figure~\ref{fig:oracle-lsh} we compare \lshrouting versus \kmeansrouting using GP as the partitioning method.
While \lshrouting performs decently, empirically \kmeansrouting consistently finds better routes which is why we excluded \lshrouting in Figure~\ref{fig:qps}.
This stems from uniform random sampling which is necessary for the theoretical analysis.

\noindent\textbf{Overlapping Partitions.}
Next, we investigate the effects of using 20\% overlap on in Figure~\ref{fig:compare-overlapping-partitions}.
Since we use more shards instead of increasing their size, it is unclear whether overlap leads to strictly better results.
Yet, we see consistent improvement of OKM over KM and OBKM over BKM, in contrast to the throughput evaluation in Figure~\ref{fig:qps}.
OGP also significantly improves upon GP on DEEP and Turing.
On TTI however, OGP has lower recall in the first shard than GP when using \kmeansrouting routing (70.6\% vs 71.7\%), even though the routing oracle achieves higher recall (85.5\% vs 81.2\%).
Overall these results demonstrate that overlapping shards can significantly boost recall.

\noindent\textbf{Losses.}
While our routing algorithms perform well, there is still a significant gap to the optimal routing oracle.
We leverage fast inexact components in two places.
We coarsen the pointset, and we use an approximate search index (HNSW) to speed up the routing query.
In the following, we thus investigate how much the inexact components contribute to the losses by replacing them with an exact version. We present the results in Figure~\ref{fig:analyzing-losses}.

We test routing based on the entire pointset to determine the effect of coarsening and study how ranking by distances loses compared to the oracle (No Coarsening).
Additionally, we test computing distances to all points in $R$ (Exact NN), to detect how much approximate search loses in routing by missing relevant points as employed in the base version of the algorithm (Approx NN).

At $\nshards=1$ probed shard, routing by ranking distances already loses significantly -- see Oracle against No Coarsening.
It is able to catch up at $\nshards=3$, suggesting that the distinction among the top 3 should be improved, but overall ranking distances is the correct approach.
At $\nshards > 1$ coarsening incurs the highest losses -- see No Coarsening vs Exact NN.
Fortunately, Approx NN incurs only a small loss against Exact NN; at $\nshards \geq 3$ on Turing and $\nshards \geq 5$ on TTI.

\subsection{Small-Scale Evaluation for Learned Partitions}\label{sec:experiments:small}

In this section, we compare with two additional baselines BLISS~\cite{bliss} and USP~\cite{fac}, which jointly learn the partition and routing index using neural networks.
Given a query, the neural net infers a probability distribution of the shards, which is interpreted as a probe order.
BLISS uses a cross-entropy loss formulation to maximize the number of points which have at least one near-neighbor in their shard, and uses the power of $k$ choices (top ranked shards) to achieve a balanced assignment.
USP uses a linear combination of edge cut and normalized squared shard size as the loss function.
Both methods rely on building a $k$-NN graph for their loss function.

Training BLISS and USP is too slow to run on the big-ann datasets, so we test on SIFT1M ($\mathrm{dim}=128$, $\ell_2$) and GLOVE ($\mathrm{dim}=100$, angular distance), which have roughly a million points, partitioned into $s=16$ shards.
Experiments are run on a 128-core EPYC 7742 clocked at 2.25GHz with 1TB RAM.
Queries are run sequentially, except for the batch-parallel results in \cref{tbl:routing-times-small-scale} which use 64 cores. Index training uses 128 cores.
We use the publicly available Python implementations in PyTorch (USP) and TensorFlow (BLISS).
The results for GP and (B)KM use \kmeansrouting routing.

\begin{table}
	\caption{Routing time in microseconds per query for different batch sizes $b$ on SIFT and GLOVE. A batch is processed in parallel.
	}\label{tbl:routing-times-small-scale}
	\vspace{0.22cm}
	\resizebox{\linewidth}{!}
	{
		\begin{tabular}{l | rrr | rrr}
			&  \multicolumn{3}{c}{SIFT} & \multicolumn{3}{c}{GLOVE} \\
			& $b=1$ & $b=32$ & $b=\infty$ & $b=1$ & $b=32$ & $b=\infty$ \\
			\midrule
			\kmeansrouting        & 95  & 13  & 2.4 & 92  & 11  & 2.1 \\
			\kmeansrouting + HNSW & 110 & 37  & 7.8 & 106 & 33  & 7.6 \\
			BLISS                 & 660 & 150 & 110 & 730 & 110 & 110 \\
			USP                   & 320 & 20  & 1.2 & 512 & 23  & 1.2 \\
		\end{tabular}
	}
\end{table}

Figure~\ref{fig:small-scale} shows the recall versus number of probed shards.
First, we observe that BLISS is completely outperformed on both datasets.
On SIFT, GP also significantly outperforms USP and (B)KM, whereas on GLOVE, USP and GP perform similarly.
This comparison uses exhaustive search in the shards, which takes 1.96ms per shard probe on SIFT. Alternatively, using HNSW takes 0.18ms per shard probe with near-equivalent recall (95.84\% vs 95.71\% in the first shard).
BLISS exhibits 13-15ms for exhaustive shard search, which is presumably an implementation issue of interfacing from Python to C++.
For USP we cannot get a clean measurement, because the shard-search is mixed with the recall calculation.
In \cref{tbl:routing-times-small-scale} we report the routing times.
USP and BLISS benefit from routing in batches, as PyTorch leverages parallelism and batched linear algebra operations.
Our approach similarly benefits from parallelism, but could be further improved via batched distance computation.
In contrast to the results on large-scale data, we observe that routing with HNSW is slower than with \kmeansrouting.

In terms of retrieval performance, USP is a good data structure.
However, it is extremely slow to train, taking 6 hours on 128 cores for GLOVE.
On the other hand BLISS is fast to train at 70s (neural network) + 566s ($k$-NN graph construction), but has poor retrieval performance.
We excluded Neural LSH~\cite{neural-lsh} from the retrieval comparison as it is already outperformed by USP, but remark that its neural net training took over two days.
Our method has the best retrieval performance at relevant batch sizes \emph{and} is also the fastest to train at just 4.76s, of which 2.95s are graph-building, 0.88s partitioning, and 0.93s \kmeansrouting.

\section{Conclusion}
We presented fast, modular and high-quality routing methods which unleash balanced graph partitioning for large-scale IVF algorithms.
Our benchmarks on billion-scale, high-dimensional data show that our methods achieve up to 2.14x higher throughput than the best baseline.
As a future work it would be interesting to explore accuracy and efficiency improvements for routing, for example exploring quantization to compress the routing index.
Another direction is to study the benefits of our approach for partitioning ANNS problems across different types of computing units, e.g. GPUs.
Additionally, we are interested in advanced partitioning cost functions that optimize locality in nearest neighbor search beyond the first shard.

%
%
\bibliography{references}

\begin{thebibliography}{42}
\providecommand{\natexlab}[1]{#1}
\providecommand{\url}[1]{\texttt{#1}}
\expandafter\ifx\csname urlstyle\endcsname\relax
  \providecommand{\doi}[1]{doi: #1}\else
  \providecommand{\doi}{doi: \begingroup \urlstyle{rm}\Url}\fi

\bibitem[Andoni et~al.(2017)Andoni, Laarhoven, Razenshteyn, and
  Waingarten]{andoni2017optimal}
Andoni, A., Laarhoven, T., Razenshteyn, I., and Waingarten, E.
\newblock Optimal hashing-based time-space trade-offs for approximate near
  neighbors.
\newblock In \emph{Proceedings of the Twenty-Eighth Annual ACM-SIAM Symposium
  on Discrete Algorithms}, pp.\  47--66. SIAM, 2017.

\bibitem[Andoni et~al.(2018)Andoni, Indyk, and
  Razenshteyn]{andoni2018approximate}
Andoni, A., Indyk, P., and Razenshteyn, I.
\newblock Approximate nearest neighbor search in high dimensions.
\newblock In \emph{Proceedings of the International Congress of Mathematicians:
  Rio de Janeiro 2018}, pp.\  3287--3318. World Scientific, 2018.

\bibitem[Babenko \& Lempitsky(2016)Babenko and Lempitsky]{babenko2016efficient}
Babenko, A. and Lempitsky, V.
\newblock Efficient indexing of billion-scale datasets of deep descriptors.
\newblock In \emph{Proceedings of the IEEE Conference on Computer Vision and
  Pattern Recognition}, pp.\  2055--2063, 2016.

\bibitem[Baranchuk \& Babenko(2021)Baranchuk and
  Babenko]{baranchuk2021benchmarks}
Baranchuk, D. and Babenko, A.
\newblock Benchmarks for billion-scale similarity search.
\newblock Webpage, 2021.
\newblock URL
  \url{https://research.yandex.com/blog/benchmarks-for-billion-scale-\\similarity-search}.

\bibitem[Baranchuk et~al.(2018)Baranchuk, Babenko, and Malkov]{ivf-hnsw}
Baranchuk, D., Babenko, A., and Malkov, Y.
\newblock {Revisiting the Inverted Indices for Billion-Scale Approximate
  Nearest Neighbors}.
\newblock In \emph{Computer Vision - {ECCV} 2018}, volume 11216 of
  \emph{Lecture Notes in Computer Science}, pp.\  209--224. Springer, 2018.
\newblock \doi{10.1007/978-3-030-01258-8\_13}.
\newblock URL \url{https://doi.org/10.1007/978-3-030-01258-8\_13}.

\bibitem[Bawa et~al.(2005)Bawa, Condie, and Ganesan]{lsh-forest}
Bawa, M., Condie, T., and Ganesan, P.
\newblock {{LSH} Forest: Self-tuning Indexes for Similarity Search}.
\newblock In Ellis, A. and Hagino, T. (eds.), \emph{Proceedings of the 14th
  international conference on World Wide Web, {WWW} 2005}, pp.\  651--660.
  {ACM}, 2005.
\newblock \doi{10.1145/1060745.1060840}.
\newblock URL \url{https://doi.org/10.1145/1060745.1060840}.

\bibitem[Bruch(2024)]{bruch2024foundations}
Bruch, S.
\newblock {Foundations of Vector Retrieval}.
\newblock \emph{arXiv preprint arXiv:2401.09350}, 2024.

\bibitem[{\c{C}}ataly{\"{u}}rek et~al.(2023){\c{C}}ataly{\"{u}}rek, Devine,
  Faraj, Gottesb{\"{u}}ren, Heuer, Meyerhenke, Sanders, Schlag, Schulz,
  Seemaier, and Wagner]{gp-survey}
{\c{C}}ataly{\"{u}}rek, {\"{U}}.~V., Devine, K.~D., Faraj, M.~F.,
  Gottesb{\"{u}}ren, L., Heuer, T., Meyerhenke, H., Sanders, P., Schlag, S.,
  Schulz, C., Seemaier, D., and Wagner, D.
\newblock {More Recent Advances in (Hyper)Graph Partitioning}.
\newblock \emph{{ACM} Computing Surveys}, 55\penalty0 (12):\penalty0
  253:1--253:38, 2023.
\newblock \doi{10.1145/3571808}.

\bibitem[Chen et~al.(2021{\natexlab{a}})Chen, Zhao, Wang, Li, Liu, Li, Yang,
  and Wang]{spann}
Chen, Q., Zhao, B., Wang, H., Li, M., Liu, C., Li, Z., Yang, M., and Wang, J.
\newblock {{SPANN:} Highly-efficient Billion-scale Approximate Nearest
  Neighborhood Search}.
\newblock In Ranzato, M., Beygelzimer, A., Dauphin, Y.~N., Liang, P., and
  Vaughan, J.~W. (eds.), \emph{Advances in Neural Information Processing
  Systems 34: NeurIPS 2021}, pp.\  5199--5212, 2021{\natexlab{a}}.
\newblock URL
  \url{https://proceedings.neurips.cc/paper/2021/hash/299dc35e747eb77177d9cea10a802da2-\\Abstract.html}.

\bibitem[Chen et~al.(2021{\natexlab{b}})Chen, Zhao, Wang, Li, Liu, Li, Yang,
  and Wang]{spann-build-times}
Chen, Q., Zhao, B., Wang, H., Li, M., Liu, C., Li, Z., Yang, M., and Wang, J.
\newblock {SPANN} paper discussion.
\newblock \url{https://openreview.net/forum?id=-1rrzmJCp4&noteId=zhMe9y8w25b},
  2021{\natexlab{b}}.
\newblock Accessed: 2023-04-14.

\bibitem[Chen et~al.(2022)Chen, Jayaram, Levi, and Waingarten]{chen2022new}
Chen, X., Jayaram, R., Levi, A., and Waingarten, E.
\newblock New streaming algorithms for high dimensional emd and mst.
\newblock In \emph{Proceedings of the 54th Annual ACM SIGACT Symposium on
  Theory of Computing}, pp.\  222--233, 2022.

\bibitem[de~Maeyer et~al.(2023)de~Maeyer, Sieranoja, and
  Fr{\"a}nti]{de2023balanced}
de~Maeyer, R., Sieranoja, S., and Fr{\"a}nti, P.
\newblock {Balanced k-means Revisited}.
\newblock \emph{Applied Computing and Intelligence}, 3\penalty0 (2):\penalty0
  145--179, 2023.

\bibitem[Deng et~al.(2019{\natexlab{a}})Deng, Yan, Ng, Jiang, and
  Cheng]{pyramid}
Deng, S., Yan, X., Ng, K. K.~W., Jiang, C., and Cheng, J.
\newblock {Pyramid: {A} General Framework for Distributed Similarity Search on
  Large-scale Datasets}.
\newblock In Baru, C.~K., Huan, J., Khan, L., Hu, X., Ak, R., Tian, Y., Barga,
  R.~S., Zaniolo, C., Lee, K., and Ye, Y.~F. (eds.), \emph{2019 {IEEE}
  International Conference on Big Data {(IEEE} BigData)}. {IEEE},
  2019{\natexlab{a}}.
\newblock \doi{10.1109/BigData47090.2019.9006219}.
\newblock URL \url{https://doi.org/10.1109/BigData47090.2019.9006219}.

\bibitem[Deng et~al.(2019{\natexlab{b}})Deng, Yan, Ng, Jiang, and
  Cheng]{pyramid-arxiv}
Deng, S., Yan, X., Ng, K. K.~W., Jiang, C., and Cheng, J.
\newblock {Pyramid: {A} General Framework for Distributed Similarity Search}.
\newblock \emph{CoRR}, abs/1906.10602, 2019{\natexlab{b}}.
\newblock URL \url{http://arxiv.org/abs/1906.10602}.

\bibitem[Dhillon \& Modha(2001)Dhillon and Modha]{SphericalKM}
Dhillon, I.~S. and Modha, D.~S.
\newblock Concept decompositions for large sparse text data using clustering.
\newblock \emph{Machine Learning}, 42\penalty0 (1/2):\penalty0 143--175, 2001.
\newblock \doi{10.1023/A:1007612920971}.

\bibitem[Dong et~al.(2020)Dong, Indyk, Razenshteyn, and Wagner]{neural-lsh}
Dong, Y., Indyk, P., Razenshteyn, I.~P., and Wagner, T.
\newblock Learning space partitions for nearest neighbor search.
\newblock In \emph{8th International Conference on Learning Representations,
  {ICLR} 2020}. OpenReview.net, 2020.
\newblock URL \url{https://openreview.net/forum?id=rkenmREFDr}.

\bibitem[Douze et~al.(2024)Douze, Guzhva, Deng, Johnson, Szilvasy, Mazar{\'e},
  Lomeli, Hosseini, and J{\'e}gou]{douze2024faiss}
Douze, M., Guzhva, A., Deng, C., Johnson, J., Szilvasy, G., Mazar{\'e}, P.-E.,
  Lomeli, M., Hosseini, L., and J{\'e}gou, H.
\newblock The faiss library.
\newblock \emph{arXiv preprint arXiv:2401.08281}, 2024.

\bibitem[Fahim et~al.(2022)Fahim, Ali, and Cheema]{fac}
Fahim, A., Ali, M.~E., and Cheema, M.~A.
\newblock {Unsupervised Space Partitioning for Nearest Neighbor Search}.
\newblock \emph{CoRR}, abs/2206.08091, 2022.
\newblock \doi{10.48550/arXiv.2206.08091}.
\newblock URL \url{https://doi.org/10.48550/arXiv.2206.08091}.

\bibitem[Fiduccia \& Mattheyses(1982)Fiduccia and Mattheyses]{fm}
Fiduccia, C.~M. and Mattheyses, R.~M.
\newblock A linear-time heuristic for improving network partitions.
\newblock In \emph{Proceedings of the 19th Design Automation Conference, {DAC}
  1982}, pp.\  175--181. {ACM/IEEE}, 1982.
\newblock \doi{10.1145/800263.809204}.
\newblock URL \url{https://doi.org/10.1145/800263.809204}.

\bibitem[Gottesb{\"{u}}ren et~al.(2021)Gottesb{\"{u}}ren, Heuer, Sanders,
  Schulz, and Seemaier]{kaminpar}
Gottesb{\"{u}}ren, L., Heuer, T., Sanders, P., Schulz, C., and Seemaier, D.
\newblock {Deep Multilevel Graph Partitioning}.
\newblock In \emph{29th Annual European Symposium on Algorithms, {ESA} 2021},
  2021.
\newblock \doi{10.4230/LIPIcs.ESA.2021.48}.
\newblock URL \url{https://doi.org/10.4230/LIPIcs.ESA.2021.48}.

\bibitem[Guo et~al.(2020)Guo, Sun, Lindgren, Geng, Simcha, Chern, and
  Kumar]{scann}
Guo, R., Sun, P., Lindgren, E., Geng, Q., Simcha, D., Chern, F., and Kumar, S.
\newblock {Accelerating Large-Scale Inference with Anisotropic Vector
  Quantization}.
\newblock In \emph{Proceedings of the 37th International Conference on Machine
  Learning, {ICML} 2020}, volume 119 of \emph{Proceedings of Machine Learning
  Research}, pp.\  3887--3896. {PMLR}, 2020.
\newblock URL \url{http://proceedings.mlr.press/v119/guo20h.html}.

\bibitem[Gupta et~al.(2022)Gupta, Medini, Shrivastava, and Smola]{bliss}
Gupta, G., Medini, T., Shrivastava, A., and Smola, A.~J.
\newblock {{BLISS:} {A} Billion scale Index using Iterative Re-partitioning}.
\newblock In Zhang, A. and Rangwala, H. (eds.), \emph{{KDD} '22: The 28th {ACM}
  {SIGKDD} Conference on Knowledge Discovery and Data Mining}, pp.\  486--495.
  {ACM}, 2022.
\newblock \doi{10.1145/3534678.3539414}.
\newblock URL \url{https://doi.org/10.1145/3534678.3539414}.

\bibitem[Indyk \& Motwani(1998)Indyk and Motwani]{indyk1998approximate}
Indyk, P. and Motwani, R.
\newblock Approximate nearest neighbors: towards removing the curse of
  dimensionality.
\newblock In \emph{Proceedings of the thirtieth annual ACM symposium on Theory
  of computing}, pp.\  604--613, 1998.

\bibitem[Jaiswal et~al.(2022)Jaiswal, Krishnaswamy, Garg, Simhadri, and
  Agrawal]{ood-diskann}
Jaiswal, S., Krishnaswamy, R., Garg, A., Simhadri, H.~V., and Agrawal, S.
\newblock {OOD-DiskANN: Efficient and Scalable Graph {ANNS} for
  Out-of-Distribution Queries}.
\newblock \emph{CoRR}, abs/2211.12850, 2022.
\newblock \doi{10.48550/arXiv.2211.12850}.
\newblock URL \url{https://doi.org/10.48550/arXiv.2211.12850}.

\bibitem[J{\'{e}}gou et~al.(2011{\natexlab{a}})J{\'{e}}gou, Douze, and
  Schmid]{douze2011product}
J{\'{e}}gou, H., Douze, M., and Schmid, C.
\newblock Product quantization for nearest neighbor search.
\newblock \emph{{IEEE} Trans. Pattern Anal. Mach. Intell.}, 33\penalty0
  (1):\penalty0 117--128, 2011{\natexlab{a}}.

\bibitem[J{\'{e}}gou et~al.(2011{\natexlab{b}})J{\'{e}}gou, Douze, and
  Schmid]{ivf-pq}
J{\'{e}}gou, H., Douze, M., and Schmid, C.
\newblock {Product Quantization for Nearest Neighbor Search}.
\newblock \emph{{IEEE} Transactions on Pattern Analysis and Machine
  Intelligence}, 33\penalty0 (1):\penalty0 117--128, 2011{\natexlab{b}}.
\newblock \doi{10.1109/TPAMI.2010.57}.
\newblock URL \url{https://doi.org/10.1109/TPAMI.2010.57}.

\bibitem[J{\'{e}}gou et~al.(2011{\natexlab{c}})J{\'{e}}gou, Tavenard, Douze,
  and Amsaleg]{jegou2011searching}
J{\'{e}}gou, H., Tavenard, R., Douze, M., and Amsaleg, L.
\newblock Searching in one billion vectors: Re-rank with source coding.
\newblock In \emph{Proceedings of the {IEEE} International Conference on
  Acoustics (ICASSP)}, pp.\  861--864. {IEEE}, 2011{\natexlab{c}}.

\bibitem[Johnson et~al.(2021)Johnson, Douze, and J{\'{e}}gou]{faiss-gpu}
Johnson, J., Douze, M., and J{\'{e}}gou, H.
\newblock {Billion-Scale Similarity Search with GPUs}.
\newblock \emph{{IEEE} Transactions on Big Data}, 7\penalty0 (3):\penalty0
  535--547, 2021.
\newblock \doi{10.1109/TBDATA.2019.2921572}.
\newblock URL \url{https://doi.org/10.1109/TBDATA.2019.2921572}.

\bibitem[Li et~al.(2019)Li, Zhang, Sun, Wang, Li, Zhang, and
  Lin]{li2019approximate}
Li, W., Zhang, Y., Sun, Y., Wang, W., Li, M., Zhang, W., and Lin, X.
\newblock Approximate nearest neighbor search on high dimensional
  data—experiments, analyses, and improvement.
\newblock \emph{IEEE Transactions on Knowledge and Data Engineering},
  32\penalty0 (8):\penalty0 1475--1488, 2019.

\bibitem[Lloyd(1982)]{Lloyd82}
Lloyd, S.~P.
\newblock Least squares quantization in {PCM}.
\newblock \emph{{IEEE} Transactions on Information Theory}, 28\penalty0
  (2):\penalty0 129--136, 1982.
\newblock \doi{10.1109/TIT.1982.1056489}.

\bibitem[Lv et~al.(2007)Lv, Josephson, Wang, Charikar, and Li]{multi-probe-lsh}
Lv, Q., Josephson, W., Wang, Z., Charikar, M., and Li, K.
\newblock {Multi-Probe {LSH:} Efficient Indexing for High-Dimensional
  Similarity Search}.
\newblock In \emph{Proceedings of the 33rd International Conference on Very
  Large Data Bases (VLDB) 2007}, pp.\  950--961. {ACM}, 2007.
\newblock URL \url{http://www.vldb.org/conf/2007/papers/research/p950-lv.pdf}.

\bibitem[Malinen \& Fr{\"{a}}nti(2014)Malinen and Fr{\"{a}}nti]{BKM-Hungarian}
Malinen, M.~I. and Fr{\"{a}}nti, P.
\newblock Balanced k-means for clustering.
\newblock In \emph{Structural, Syntactic, and Statistical Pattern Recognition -
  Joint {IAPR} International Workshop, {S+SSPR} 2014, Joensuu, Finland, August
  20-22, 2014. Proceedings}, volume 8621 of \emph{Lecture Notes in Computer
  Science}, pp.\  32--41. Springer, 2014.
\newblock \doi{10.1007/978-3-662-44415-3\_4}.

\bibitem[Malkov \& Yashunin(2020)Malkov and Yashunin]{hnsw-journal}
Malkov, Y.~A. and Yashunin, D.~A.
\newblock {Efficient and Robust Approximate Nearest Neighbor Search Using
  Hierarchical Navigable Small World Graphs}.
\newblock \emph{{IEEE} Transactions on Pattern Analysis and Machine
  Intelligence}, 42\penalty0 (4):\penalty0 824--836, 2020.
\newblock \doi{10.1109/TPAMI.2018.2889473}.
\newblock URL \url{https://doi.org/10.1109/TPAMI.2018.2889473}.

\bibitem[Muja \& Lowe(2014)Muja and Lowe]{flann}
Muja, M. and Lowe, D.~G.
\newblock {Scalable Nearest Neighbor Algorithms for High Dimensional Data}.
\newblock \emph{{IEEE} Transactions on Pattern Analysis and Machine
  Intelligence}, 36\penalty0 (11):\penalty0 2227--2240, 2014.
\newblock \doi{10.1109/TPAMI.2014.2321376}.
\newblock URL \url{https://doi.org/10.1109/TPAMI.2014.2321376}.

\bibitem[Nist{\'{e}}r \& Stew{\'{e}}nius(2006)Nist{\'{e}}r and
  Stew{\'{e}}nius]{k-means-tree}
Nist{\'{e}}r, D. and Stew{\'{e}}nius, H.
\newblock {Scalable Recognition with a Vocabulary Tree}.
\newblock In \emph{2006 {IEEE} Computer Society Conference on Computer Vision
  and Pattern Recognition {(CVPR} 2006)}, pp.\  2161--2168. {IEEE} Computer
  Society, 2006.
\newblock \doi{10.1109/CVPR.2006.264}.
\newblock URL \url{https://doi.org/10.1109/CVPR.2006.264}.

\bibitem[Paz(2014)]{InfiniBand}
Paz, O.
\newblock Infini{B}and essentials every {HPC} expert must know.
\newblock HPC Advisory Council Switzerland Conference, 2014.
\newblock URL
  \url{https://www.hpcadvisorycouncil.com/events/2014/swiss-workshop/presos/Day_1/1_Mellanox.pdf}.

\bibitem[Pennington et~al.(2014)Pennington, Socher, and
  Manning]{pennington2014glove}
Pennington, J., Socher, R., and Manning, C.~D.
\newblock Glove: Global vectors for word representation.
\newblock In \emph{Proceedings of the 2014 conference on empirical methods in
  natural language processing (EMNLP)}, pp.\  1532--1543, 2014.

\bibitem[Russell \& Norvig(2009)Russell and Norvig]{norvig-russell}
Russell, S.~J. and Norvig, P.
\newblock \emph{Artificial Intelligence: a modern approach}.
\newblock Pearson, 3 edition, 2009.

\bibitem[Shakhnarovich et~al.(2008)Shakhnarovich, Darrell, and
  Indyk]{shakhnarovich2008nearest}
Shakhnarovich, G., Darrell, T., and Indyk, P.
\newblock Nearest-neighbor methods in learning and vision.
\newblock \emph{IEEE Trans. Neural Networks}, 19\penalty0 (2):\penalty0 377,
  2008.

\bibitem[Simhadri et~al.(2021)Simhadri, Williams, Aum{\"{u}}ller, Douze,
  Babenko, Baranchuk, Chen, Hosseini, Krishnaswamy, Srinivasa, Subramanya, and
  Wang]{big-ann-benchmarks}
Simhadri, H.~V., Williams, G., Aum{\"{u}}ller, M., Douze, M., Babenko, A.,
  Baranchuk, D., Chen, Q., Hosseini, L., Krishnaswamy, R., Srinivasa, G.,
  Subramanya, S.~J., and Wang, J.
\newblock {Results of the NeurIPS'21 Challenge on Billion-Scale Approximate
  Nearest Neighbor Search}.
\newblock In Kiela, D., Ciccone, M., and Caputo, B. (eds.), \emph{NeurIPS 2021
  Competitions and Demonstrations Track}, volume 176 of \emph{Proceedings of
  Machine Learning Research}, pp.\  177--189. {PMLR}, 2021.
\newblock URL \url{https://proceedings.mlr.press/v176/simhadri22a.html}.

\bibitem[Subramanya et~al.(2019)Subramanya, Devvrit, Simhadri, Krishnaswamy,
  and Kadekodi]{diskann}
Subramanya, S.~J., Devvrit, F., Simhadri, H.~V., Krishnaswamy, R., and
  Kadekodi, R.
\newblock {DiskANN: Fast Accurate Billion-point Nearest Neighbor Search on a
  Single Node}.
\newblock In Wallach, H.~M., Larochelle, H., Beygelzimer, A.,
  d'Alch{\'{e}}{-}Buc, F., Fox, E.~B., and Garnett, R. (eds.), \emph{Advances
  in Neural Information Processing Systems 32: Annual Conference on Neural
  Information Processing Systems 2019}, pp.\  13748--13758, 2019.
\newblock URL
  \url{https://proceedings.neurips.cc/paper/2019/hash/09853c7fb1d3f8ee67a61b6bf4a7f8e6-\\Abstract.html}.

\bibitem[Sun et~al.(2023)Sun, Simcha, Dopson, Guo, and Kumar]{sun2023soar}
Sun, P., Simcha, D., Dopson, D., Guo, R., and Kumar, S.
\newblock Soar: Improved indexing for approximate nearest neighbor search.
\newblock In \emph{Thirty-seventh Conference on Neural Information Processing
  Systems}, 2023.

\end{thebibliography}
\bibliographystyle{icml2024}

\newpage
\appendix
\onecolumn

\section{\lshrouting pseudocodes}\label{app:codes}
\begin{algorithm2e}
	\caption{\lshrouting: Training}\label{algo:lsh:build}
	\KwIn{Shards $S_i$, index size $m \leq n$, repetition parameter $r$, sketch size $t$}
	\ParallelFor() {$i=1$ to $s$}{
		Sample $\lfloor \frac{m \cdot |S_i|}{|P|} \rfloor$ points $R_i$ from $S_i$ uniformly without replacement \;
	}
	$R = \cup_{i=1}^s R_i$\;
	\ParallelFor() {$j=1$ to $r$}{
		Draw independent random hash functions $h_{j,1},\dots,h_{j,t}$ from a LSH family for $\cX$ \;
		For each $x \in R$ define a sorting key $h_j(x) = (h_{j,1}(x),\dots,h_{j,t}(x))$

		Sort $R$ lexicographically with keys $h_j$ to obtain $I_j$

	}

\end{algorithm2e}
\begin{algorithm2e}
	\caption{\lshrouting: Routing}\label{algo:lsh:routing}
	\KwIn{Query point $q$, window size $W\geq 1$,  partition $\Partition' \colon R \mapsto [s]$ }
	min-dist$[i] \gets \infty \quad \forall i \in [s]$ \;
	\For() {$j=1$ to $r$ } {
		Binary search in $I_j$ for position $\tau$ of the point lexicographically closest to $h_j(q)$ \;
		\For() {$c \in I_j[\tau - W \colon \tau + W]$} {
			$\text{min-dist}[\Partition'[c]] \gets \min(\text{min-dist}[\Partition'[c]], d(q,c))$ \;
		}
	}
	\KwRet sort $[s]$ by min-dist \;

\end{algorithm2e}

\section{Routing via Voting}\label{sec:appendix:voting}

In this section, we present a method called \voting, which is an alternative to ranking shards by the distance of the coarse representative of the shard that is closest to the query, which we call \ranking in the following.

Let $Q$ be the set of relevant coarse representatives retrieved from the routing index for a query $q \in \cX$.
With \ranking the shards are probed in order of increasing minimum distance $\min_{v \in R_i}(d(q,v))$.
We expect many of $q$'s neighbors to be concentrated in the same shard because we optimized the partition for specifically this metric.
Therefore, routing to the shard of any \emph{near} neighbor is likely good; we use the closest one we find.

The following example demonstrates a flaw of \ranking, when this intuition does not hold.
If shard $S_a$ has 99 of the top-100 neighbors but $S_b$ has the top-$1$ neighbor (and assuming the routing index finds it) then \ranking inspects $S_b$ first whereas inspecting $S_a$ first would be better.
In the \voting scheme all points in $Q$ influence the outcome but their influence decays with increasing distance to $q$.
This prefers many \emph{near} neighbors over a single one, without using a hard cutoff (top few) to determine what near means.
For each shard $S_i$ we compute its voting power $\nu(S_i) \coloneqq \sum_{v \in R_i} e^{-\sigma d(q,v)^2}$.
Here $\sigma$ is a scaling factor to remap the distances to the interval $[0, 12]$ (an experimentally determined cutoff).
We probe shards with higher voting power first.

\clearpage

\section{Proof of Theorem 1}\label{sec:appendix:sorting-lsh}

In this section, we provide the missing proof for \cref{thm:SortLSHMain}. We first restate the theorem to accomodate both \ranking and \voting.

\addtocounter{theorem}{-1}
\begin{theorem}
Fix any approximation factor $c>1$, and let $P \subset (\R^d,\|\cdot\|_\rho)$ be a subset of the $d$-dimensional space equipped with the $\ell_\rho$ norm, for any $\rho \in [1,2]$. Set $\alpha = O( \sqrt{c\log (n ) })$ for the algorithm using \voting, and $\alpha = O(c)$ for the algorithm using \ranking. Then there exists a locally sensitive hash family $\cH$, such that with $r = O(n^{1/c})$ and window size $W = O(1)$, the following holds when running the SortingLSH routing algorithm: for any query point $q \in \R^d$,
\begin{itemize}
    \item If $m=n$, then with probability $1-1/\poly(n)$ $q$ gets routed to a shard $\bPi(p)$ containing at least one point $p \in P$ that satisfies
    $$ \|p - q \|_\rho \leq \alpha \pi_1(q) $$
    \item if $m<n$, then for any $\delta \in (0,1)$, with probability $1-\delta$ the query $q$ gets routed to a shard $\bPi(p)$ containing at least one point $p \in P$ that satisfies
    $$ \|p - q \|_\rho \leq \alpha \cdot \pi_{\lceil \log\delta^{-1}  \frac{n}{m}\rceil  }(q) $$
\end{itemize}
Where $\pi_k(q)$ is the distance from $q$ to the $k$-th nearest neighbor of $q$ in $P$.

\end{theorem}

\begin{proof}

To simplify the construction of the hash family $\cH$, we first embed the pointset $P$ into a subset of the $d'$-dimensional Hamming cube $\{0,1\}^{d'}$ with a constant distortion in distances.  Let $\Phi = \frac{\max_{x,y \in P} \|x-y\|_\rho}{\min_{x,y \in P} \|x-y\|_\rho}$ denote the aspect ratio of $P$.  By Lemma A.2 and A.3 of \cite{chen2022new}, for any $p \in [1,2]$ such an embedding $f_\rho:P \to \{0,1\}^{d'}$, with $d' = O(d \Phi \log n)$, such that there exists a constant $C$ so that with probability $1-1/\poly(n)$ for all $x,y \in P$, we have
\[ \|f_\rho(x) - f_\rho(y)\|_0  \leq  \|x-y\|_\rho  \leq C  \cdot \|f_\rho(x) - f_\rho(y)\|_0 \]
Where  $\|a-b\|_0 = |\{ i \in [d'] | a_i \neq b_i\}|$ is the Hamming distance between any two $a,b \in \{0,1\}^{d'}$. The constant $C$ will go into the approximation factor of the retrieval, but note that $C$ can be set to $(1+\eps)$ by increasing the dimension by a $O(1/\eps^2)$ factor. Thus, in what follows, we may assume that $P \subset \{0,1\}^{d'}$ is a subset of the $d'$-dimensional hypercube equipped with the Hamming distance.

We now construct the hash family $\cH$ which we will use for the SortingLSH Index. A draw $h \sim \cH$ is generated as follows: (1) First, sample $i_1,\dots,i_\ell \sim [d']$ uniformly at random, where $\ell = O(d' \log n)$, (2) for any $x \in \{0,1\}^{d'}$ we define $h(x) = (x_{i_1}, x_{i_2}, \dots, x_{i_{\ell}}) \in \{0,1\}^{\ell}$. As notation, for any $\ell' \leq \ell$ and hash function $h$ defined in the above way, we write $h_{\ell'}(x) =(x_{i_1}, x_{i_2}, \dots, x_{i_{\ell'}}) \in \{0,1\}^{\ell'}$ to denote the $\ell'$-prefix of $h$.

We are now ready to prove the main claims of the Theorem.
First, suppose we are in the case that $m=n$, and thus the pointset is \textit{not} subsampled before the construction of the SortingLSH index.
Let $p^* = \arg \min_{p \in R} \|p'-q\|_0$ be the nearest neighbor to $q$, with ties broken arbitrarily.
We first claim that $p^*$ and $q$ share a $t$-length prefix in at least one of the $i \in [r]$ repetitions of SortingLSH. Set $t^* = \frac{4 d' \ln n }{ c \|p^*-q\|_0}$.
Then the probability that $p^*,q$ share a prefix of length at least $t^*$ -- namely, the event that $h_{t^*}(p^*) = h_{t^*}(q)$, is at least
\[ \left(1 - \frac{\|p^* - q\|_0}{d'} \right)^{t^*} =   \left(1 - \frac{\|p^* - q\|_0}{d'} \right)^{\frac{(4/c) d' \log n }{\|p^*-q\|_0}} \geq \left(\frac{1}{2}\right )^{ (8/c)  \log n }= n^{O(1/c)} \]
where we used the inequality $(1-x/2)^{1/x} \geq 1/2$ for any $0 < x < 1$. Next, for any $x$ such that $\|x-q\|_0 \geq 10/c \|p^* - q\|_0$, we have:

\[ \left(1 - \frac{\|x - q\|_0}{d'} \right)^{t*} =   \left(1 - \frac{\|x - q\|_0}{d'} \right)^{\frac{(1/c) d' \log n }{\|p^*-q\|_0}} \leq 1/n^4\]

where we used the inequality that $(1-x)^{n/x} \leq (1/2)^n$ for any $x \in (0,1]$ and $n \geq 1$. Union bounding over all $r < n$ trials, it follows that with probability at least $1-1/n^3$, we never have $h_{t^*}(x) = h_{t^*}(q)$ for any $x$ such that $\|x-q\|_0 \geq 10/c \|p^* - q\|_0$, and moreover we do have that $h_{t^*}(p^*) = h_{t^*}(q)$ for at least one repetition of the sampling.

Let $h^1,h^2,\dots,h^r$ be the $r$ independent hash functions drawn for the SortingLSH routing index. By the above, there exists a repetition $i \in [r]$ such that $h^i_{t^*}(p^*) = h^i_{t^*}(q)$.

We first prove the bounds for \ranking. First, suppose that $p^*$ was added to the window. Then since $p^*$ is the closest point to $q$, by construction of the algorithm the point $q$ will be deterministically routed to $\bPi[p']$ for a point $p'$ such that $\|q-p'\|_\rho = \|q-p^*\|_\rho$, which completes the proof in this case. Otherwise, on repetition $i$, if $p^*$ was not added to the window, then since $h^i_{t^*}(p^*) = h^i_{t^*}(q)$, there must be another point $p'$ with $h^i_{t^*}(p') = h^i_{t^*}(q)$ on that repetition. By the above, such a point must satisfy $\|q-p'\|_\rho\leq O(c) \|q-p^*\|_\rho$ as needed.

We now consider the case of \voting. Again, for the same step $i$, we recover at least one point $p'$ such that $h^i_{t^*}(p') = h^i_{t^*}(q)$, and thus $\|q-p'\|_\rho\leq  \|q-p^*\|_\rho$. Normalize the votes so that $e^{-\sigma^2 \|p'-q\|_\rho^2} = 1$. Then note that any point $p''$ with $\|p''-q\|_\rho > O(\sqrt{ \log rW }) \|p'-q\|_\rho$ will contribute a total voting weight of at most $\frac{1}{2rW}$.
Since there are at most $rW$ such points, the total voting weight to points to shards potentially not containing a $\alpha$-nearest neighbor, for $\alpha = (c \sqrt{\log Rw }) = O(c \sqrt{\log n^{1/c}})$,  is at most $1/2$. It follows that $q$ is routed to a shard which contains a $O(\alpha)$-approximate nearest neighbor as needed.

Finally, the case of $m<n$ follows by noting that the $O(\log(1/\delta)n/m)$-th nearest neighbor to $q$ will survive after sub-sampling $m$ points from $n$ with probability at least $1-\delta/2$, in which case the above analysis applies.

\end{proof}

\section{Datasets}
We utilize three billion-size datasets from the  big-ann-benchmarks competition framework~\cite{big-ann-benchmarks}.
The \emph{DEEP} (DEEP) dataset released by Yandex consists of a billion image descriptors of the
projected and normalized outputs from the last fully-connected layer of
the GoogLeNet model, which was pretrained on the Imagenet classification task~\cite{babenko2016efficient}. DEEP uses the $\ell_2$ distance.
The \emph{Turing} dataset released by Microsoft consists of 1B Bing queries encoded by the Turing v5 architecture
that trains  Transformers to capture similarity of intent in web search queries; the dataset
was released as part of the big-ann benchmarks competition in 2021~\cite{big-ann-benchmarks}.
Turing also uses the $\ell_2$ distance.
The \emph{Text-to-Image} dataset released by Yandex Research, consists of a set of images
embedded using the SeResNext-101 model, and a set of textual queries
embedded using a DSSM model. Its vectors are represented using 4 byte
floats in 200 dimensions~\cite{baranchuk2021benchmarks}. Text-to-Image uses inner product as the similarity measure.

We also report some results on the SIFT1M and GLOVE datasets. The SIFT dataset which
consists of SIFT image similarity
descriptors applied to images~\cite{big-ann-benchmarks, jegou2011searching, douze2011product}.
It is encoded as 128-dimensional vectors using 4 byte floats per vector entry, and uses the $\ell_2$ distance.
The GLOVE dataset consists of 1.18M 100-dimensional word embeddings, equipped with angular distance~\cite{pennington2014glove}.

\section{Baselines}\label{appendix:baselines}

In this appendix, we provide additional information on how the baselines $k$-means clustering (KM), balanced $k$-means (BKM) clustering~\cite{de2023balanced} and Pyramid~\cite{pyramid} are implemented.
Additionally, we discuss how we adapt $k$-means clustering to datasets with inner product as the similarity measure, as the standard algorithm works primarily for $l_2$ distance.

\subsection{K-Means}

Our $k$-means implementation is a standard version of LLoyd's algorithm~\cite{Lloyd82} with 20 rounds and random points from the pointset as initial centroids.

\paragraph{Balancing.}

While $k$-means clusters are already fairly balanced, they often do not adhere to the strict shard size limit imposed to satisfy memory constraints.
For this baseline, we tested two approaches to rebalance a $k$-means clustering: 1) remigrating points from overloaded clusters to their second closest center (etc.) and 2) splitting overloaded clusters recursively with $k$-means.
Remigrating points from heavy clusters works better for QPS but worse for partition quality (recall with the routing oracle). Splitting heavy clusters works better for partition quality but worse for QPS because splitting a cluster takes away an available replica to counteract load imbalance.
In the experiments we opted for reporting results with remigrating points since the focus is on the throughput evaluation.

\paragraph{Inner product similarity.}

By default $k$-means optimizes for squared $\ell_2$ distance within clusters.
The cluster assignment step in LLoyd's algorithm can still be used to optimize for inner product similarity.
The centroid aggregation step however differs.
We follow the spherical $k$-means approach~\cite{SphericalKM} of normalizing centroids after each iteration.
More precisely, the best performing approach in preliminary tests was to rescale each centroid to the average norm in its cluster.

In the experiments we use one dataset with inner product search (TTI) and one with angular distance (GloVe).
While TTI is not normalized, the vector norms are all within a factor of roughly 2 from each other.
ANNS with angular distance corresponds to maximum inner product search when the points are normalized.
Therefore, the above approach correctly rescales centroids to unit norm for the case of angular distance.
For a more detailed discussion on how to approach maximum inner product search, we refer to~\cite{sun2023soar} and~\cite{douze2024faiss}.

\subsection{Balanced $k$-means}

Our implementation of balanced $k$-means is based on a recent paper~\cite{de2023balanced} and their publicly available implementation at \url{https://github.com/uef-machine-learning/Balanced_k-Means_Revisited}, including their parameter choices.
This method is more scalable than for example the well-known balanced $k$-means with the Hungarian method~\cite{BKM-Hungarian}.

We use standard $k$-means with Lloyd's algorithm to initialize the clusters.
As long the clustering is not balanced, we then perform iterations of the algorithm by~\cite{de2023balanced}.
Essentially, it is a version of Lloyd's assignment algorithm with a penalty term on heavy clusters added to the cost function.
The improvement over previous methods is that the trade-off between the $k$-means objective and balance is automatically tuned over the course of the execution, and no longer in the hands of the user.
In each round the penalty is carefully adjusted to ensure that at least a small amount of points remigrate, to ensure progress towards balance, while preserving the cluster structure as much as possible.

We parallelize BKM by moving points to their preferred cluster in small synchronous sub-rounds in parallel.
The centroids are only updated after each sub-round.
To keep the centroids representative of their cluster, we use a large number of sub-rounds (1000).

\subsection{Pyramid}\label{sec:appendix:pyramid-imbalanced}

In the following we describe the partitioning and routing used in Pyramid~\cite{pyramid}.
First the dataset $P$ is randomly sub-sampled to a smaller pointset $P'$.
$P'$ is then further aggregated via flat $k$-means to $\hat{P}$.
On $\hat{P}$ a (HNSW) graph is built, which is partitioned into balanced shards using a graph partitioning algorithm.
This HNSW graph is used as the routing index -- all shards with points visited in a query are probed.
Therefore, the beamwidth of the search affects the number of probes.
The partition of $\hat{P}$ is extended to $P$ by assigning each point to the shard of its nearest neighbor in $\hat{P}$.
Therefore, the partitioning method is directly tied to the routing method.

The source code of Pyramid is not available.
Fortunately, their method is simple such that we could reimplement it in our framework, using the same parameters as in the paper~\cite{pyramid-arxiv} ($|\hat{P}| = 10000$).
Note that we achieved better partitions and query throughput by building and partitioning a $k$-NN graph on $\hat{P}$ rather than partitioning the HNSW graph. For routing we still use the HNSW graph built on $\hat{P}$.

Because of the last assignment step, the shards are highly imbalanced.
In Table~\ref{tbl:pyramid-imbalance} we report that Pyramid exhibits between $24.8\% - 88.8\%$ imbalance, whereas our method with graph partitioning achieves the desired $5\%$ imbalance, while having $0.7\% - 10.2\%$ better recall in the first shard with the routing oracle.
If we enforce a $5\%$ imbalance for Pyramid by reassigning points to the closest cluster below the size constraint, recall drops by $0.7\% - 4.1\%$.
To keep the comparison fair, we use this balanced version in the experiments.

\begin{table}[H]
	\centering
	\caption{Imbalance and first shard recall with the routing oracle for our method and Pyramid.}\label{tbl:pyramid-imbalance}
	\vspace{0.22cm}
	\begin{tabular}{l l r r}
		dataset & algorithm & max shard size & first shard recall / routing oracle \\ \midrule
		DEEP & GP & 26.25M & 85\% \\
		DEEP & Pyramid & 33.4M & 81.3\% \\
		DEEP & Pyramid balanced & 26.25M & 77.2\% \\
		\midrule
		Turing & GP & 26.25M & 86.1\% \\
		Turing & Pyramid & 31.2M & 75.9\% \\
		Turing & Pyramid balanced & 26.25M & 75.2\% \\
		\midrule
		Text-to-Image & GP & 26.25M & 81.2\% \\
		Text-to-Image & Pyramid & 47.2M & 80.5\% \\
		Text-to-Image & Pyramid balanced & 26.25M & 76.9\% \\

	\end{tabular}
\end{table}

\section{Algorithm Configuration}\label{appendix:configuration}

In this appendix, we present the parameters used in our evaluation.

\subsection{Graph-Building Configuration}

For $k$-NN graph building with dense ball clusters we perform three independent repetitions.
We assign points to the three closest pivots on the top recursion level, and to the one closest pivot on levels below.
We call this parameter \emph{fanout} (only used on the top level).
As maximum cluster size we use 2500.
The number of pivots is set as a constant on the top recursion level (950) and as a fraction of the remaining number of points (0.005) on the lower levels.

For the experiments in Figure~\ref{plot:graph_quality}, we sweep through all combinations of the following values: repetitions $\in \{2,3,5,8,10\}$, fanout $\in \{2,3,5,8,10\}$, and cluster size $\in \{ 500, 1000, 2000, 5000, 10000\}$ to obtain graphs with different quality scores.

\subsection{Configuration for Big-ANN benchmarks}

For \kmeansrouting on big-ann benchmarks with $|P|=1B$ points, we use a cluster size threshold of $\lambda = 350$, number of centroids $l = 64$ and tree search budget $b=50K$.
In preliminary experiments, we tested different parameter settings and found the results were not sensitive to $\lambda$ and $l$ which is why these results are excluded here, whereas $b$ does influence routing quality.
Note that we do not explore different budgets $b$ here, as \kmeansrouting with tree-search is only used in Section~\ref{sec:experiments:perfect-inshard-query}.
For the throughput evaluation we use HNSW instead to achieve faster routing times.
To explore different routing quality configurations, we therefore vary the size $m$ of the set of coarse reprentatives $R$, testing $m \in \{ 20K, 100K, 200K, 500K, 1M, 2M, 5M, 10M \}$.

The HNSW configuration for routing uses a degree of $32$ in the search-graph and beamwidth $\text{ef\_construction} = 200$ for insertion.
To explore different routing quality trade-offs we vary the beamwidth during search $\text{ef\_search} \in \{20,40,80,120,200,250,300,400,500 \}$.
For $m=5M$ the routing latency on MS Turing with OGP ranges from 0.11ms to 1.19ms, reaching 64.8\% to 81.7\% of ground-truth neighbors in the rop-ranked shard.

Note that each tested parameter value appears in at least one Pareto-optimal configuration.
The best performing configurations on the Pareto front may combine a low-quality choice for $m$ and a high-quality choice for $\text{ef\_search}$ or vice versa.

For in-shard search, we also use degree $32$ and $\text{ef\_construction} = 200$.
For different in-shard search recall trade-offs we vary $\text{ef\_search} \in \{50,80,100,150,200,250,300,400,500 \}$.

Given these configurations, we can analyze the memory footprint per host.
On the MS Turing dataset with 100-dimensional 4 byte vectors, the total memory required to host the 26.25M points in a shard is roughly $12.9$GB, consisting of $9.78$GB $ = 26.25 \cdot 10^6 \cdot 100 \cdot 4$ bytes for the vectors plus $3.12$GB $ = 26.25 \cdot 10^6 \cdot 32 \cdot 4$ bytes for the edges of the HNSW search graph.
Similarly the routing index consumes between 10MB (for $m=20K$) and $4.9$GB (for $m=10^7$) of memory.
In total the maximum memory per host is thus $17.8$GB, fitting comfortably into typical cluster machines which have around 100GB of memory.

For routing with \lshrouting we test the number of repetitions $r \in \{ 1,4,8,16,24 \}$ and window size $W \in \{ 50,  100,  200,  400, 1000, 2000, 4000 \}$. We use the same values for $m$ as with \kmeansrouting.
The results in Figure~\ref{fig:oracle-lsh} use the best performing configuration with search budget $b = r \cdot 2W \leq 50K$.

\subsection{Configuration for SIFT and GloVe}

On SIFT and GloVe we use a smaller router and HNSW graph with \kmeansrouting.
We set the router size $m=50K$, the search budget $b=5K$, number of centroids $l=32$ and cluster size threshold $\lambda =200$.
The HNSW graphs use degree $16$ and $\text{ef\_construction} = 200$. For routing we use $\text{ef\_search} = 60$, whereas for in-shard search we use $\text{ef\_search} = 120$.

\end{document}